\documentclass[12pt]{amsart}
\usepackage{amssymb}
\usepackage{color}
\usepackage{amsmath,epic,curves}
\usepackage[english]{babel}
\usepackage{graphicx}
\newtheorem{claim}{}[section]
\newtheorem{theorem}[claim]{Theorem}

\newtheorem{proposition}[claim]{Proposition}
\newtheorem{corollary}[claim]{Corollary}

\renewenvironment{proof}{\noindent{\it Proof. \hskip0pt}}
                      {$\square$\par\medskip}

\begin{document}
\baselineskip 6.0 truemm
\parindent 1.5 true pc

\newcommand\lan{\langle}
\newcommand\ran{\rangle}
\newcommand\tr{{\text{\rm Tr}}\,}
\newcommand\ot{\otimes}
\newcommand\ol{\overline}
\newcommand\join{\vee}
\newcommand\meet{\wedge}
\renewcommand\ker{{\text{\rm Ker}}\,}
\newcommand\image{{\text{\rm Im}}\,}
\newcommand\id{{\text{\rm id}}}
\newcommand\tp{{\text{\rm tp}}}
\newcommand\pr{\prime}
\newcommand\e{\epsilon}
\newcommand\la{\lambda}
\newcommand\inte{{\text{\rm int}}\,}
\newcommand\ttt{{\text{\rm t}}}
\newcommand\spa{{\text{\rm span}}\,}
\newcommand\conv{{\text{\rm conv}}\,}
\newcommand\rank{\ {\text{\rm rank of}}\ }
\newcommand\re{{\text{\rm Re}}\,}
\newcommand\ppt{\mathbb T}
\newcommand\rk{{\text{\rm rank}}\,}
\newcommand\bcolor{\color{blue}}
\newcommand\ecolor{\color{black}}
\newcommand\sss{\omega}

\title{Separable states with unique decompositions}

\author{Kil-Chan Ha and Seung-Hyeok Kye}
\address{Faculty of Mathematics and Applied Statistics, Sejong University, Seoul 143-747, Korea}
\address{Department of Mathematics and Institute of Mathematics\\Seoul National University\\Seoul 151-742, Korea}

\date\today

\thanks{partially supported by NRFK 2012-0002600 and NRFK 2012-0000939}



\begin{abstract}
We search for faces of the convex set consisting of
all separable states, which are affinely isomorphic to simplices,  to get separable states with unique decompositions.
In the two-qutrit case, we found that six product vectors spanning a five
dimensional space give rise to  a face isomorphic to the
5-dimensional simplex with six vertices, under suitable linear independence assumption.
 If the partial conjugates of six product vectors also span a 5-dimensional space, then
this face is inscribed in the face for PPT states whose boundary
shares the fifteen 3-simplices on the boundary of the 5-simplex. The
remaining boundary points consist of PPT entangled edge states of
rank four. We also show that every edge state of rank four arises in
this way.  If the partial conjugates of the above six product
vectors span a 6-dimensional space then we have a face isomorphic to
5-simplex, whose interior consists of separable states with unique
decompositions, but with non-symmetric ranks. We also construct a
face isomorphic to the 9-simplex. As applications, we give answers
to questions in the literature
\cite{chen_dj_semialg,chen_dj_ext_PPT}, and construct $3\ot 3$ PPT
states of type (9,5). For the qubit-qudit cases with $d\ge 3$, we
also show that $(d+1)$-dimensional subspaces give rise to faces
isomorphic to the $d$-simplices, in most cases.
\end{abstract}

\maketitle

\section{Introduction}\label{sec:intro}

The notion of quantum entanglement is one of the main current
research subjects in quantum information theory as well as quantum
physics, where the fundamental question is, of course, how to
distinguish entangled states from separable ones  \cite{werner}. One of useful
criteria is to use the notion of partial transpose, as it was observed
by Choi \cite{choi-ppt} and Peres \cite{peres}: The partial
transpose of a separable state is still  positive, that is, positive semi-definite,  but
the converse is not true. Even though we have a more complete
criterion using positive linear maps \cite{eom-kye,horo-1}, it is
not so easy to apply this criterion due to the current stage of poor
understanding on positive linear maps.

We denote by $\mathbb S$ the convex set of all separable states on
$M_m\ot M_n$, where $M_m$ denotes the algebra of all $m\times m$
matrices over the complex field. A state is called separable if it is a
convex combination of pure product states. Therefore, every
separable state is of the form
\begin{equation}\label{sep}
\varrho=\sum_{i=1}^k \lambda_i |x_i\ot y_i\rangle\langle x_i\ot y_i|,
\end{equation}
with unit   product
vectors $|x_i\ot y_i\rangle:=|x_i\rangle \otimes |y_i\rangle \in\mathbb C^m\ot \mathbb C^n$ and positive numbers $\lambda_i$ satisfying $\sum_{i=1}^k\lambda_i=1$.
A non-separable state is called entangled.
The partial transpose $\varrho^\Gamma$ of a state $\varrho$ is obtained by taking the transpose in the first component. Then the partial
transpose of a separable state in (\ref{sep}) is given by
$$
\varrho^\Gamma=\sum_{i=1}^k \lambda_i |\bar x_i\ot y_i\rangle\langle \bar x_i\ot y_i|,
$$
which is still positive, where $|\bar x\rangle$ denotes the vector whose entries are the complex conjugates of the corresponding
entries of the vector $|x\rangle$. The product vector $|\bar x\ot y\rangle$ is called the
partial conjugate of $|x\ot y\rangle$.
 We denote by $\mathbb T$ the convex set of all states whose partial transposes are positive.
 The positive partial transpose (PPT) criterion tells us that the relation
$\mathbb S\subset\mathbb T$ holds. These two convex sets coincide if and only if $mn\le 6$,
as it was shown by Woronowicz \cite{woronowicz},
Choi \cite{choi-ppt} and Horodecki's \cite{horo-1}.

We note that the boundary between separable states and entangled
ones consists of faces of the convex set $\mathbb S$, and so it is
very important to understand the whole facial structures of $\mathbb
S$. First of all, we note that every extreme point of $\mathbb S$ is
a pure product state $|x\ot y\rangle\langle x\ot y|$. The next step
to understand the facial structures of $\mathbb S$ is naturally to
find faces of $\mathbb S$ which are affinely isomorphic to
simplices. We call those {\sl simplicial faces}. We recall
that a convex subset $F$ of a convex set $C$ is said to be a face of
$C$ is the following property hold:
$$
x,y \in C,\, (1-t)x+t y\in F \text{ for some } t\in (0,1)\ \Longrightarrow \ x, y \in F.
$$
We note that every point $p$ in a convex set determines a unique
face in which $p$ is an interior point with respect to the relative topology
given by the affine manifold generated by the face. This is the smallest face containing the point $p$.
By definition of face, it is clear that the following two statements for a separable state $\varrho$ of the form
(\ref{sep}) are equivalent:
\begin{itemize}
\item
$\varrho$ determines a simplicial face.
\item
The decomposition (\ref{sep}) of $\varrho$ is unique.
\end{itemize}

It was shown by Kirkpatrick \cite{kirk} that if $\{|x_i\rangle\}$ is
pairwise distinct up to constant multiplications and
$\{|y_i\rangle\}$ is linearly independent then the decomposition
(\ref{sep}) is unique. This result was also found by Alfsen and
Shultz \cite{alfsen} independently in the context of facial
structures. See also \cite{alfsen_2} for further progresses.
Recently, Chen and Djokovi\'c \cite{chen_dj_ext_PPT} extended this result to show
that if product vectors in (\ref{sep}) with $k\le m+n-2$ satisfy a suitable linear independence condition for each subsystem
then they are linearly independent and $|x_i \otimes
y_i\rangle$ are the only product vectors in the span of them.
This clearly implies that the separable states determined by these vectors
have unique decompositions.
We note that this result covers the cases when the ranks of
separable states do not exceed $m+n-2$. The main purpose of this
paper is to search for separable states with unique decompositions
whose ranks exceed the number $m+n-2$.

In the two-qutrit case, it is well known
that generic five dimensional subspaces of $\mathbb C^3\otimes\mathbb C^3$
have exactly six product vectors $\{|x_i\otimes y_i\rangle\}$, and their orthogonal complements have no product vectors.
We show that these six product vectors give rise to
a simplicial face isomorphic to the $5$-simplex $\Delta_5$ whose interior consists of separable states of rank five. The
interiors of the six maximal faces,
isomorphic to the $4$-simplex $\Delta_4$, also consist of separable states of rank five.
The interiors of the next fifteen maximal faces, isomorphic to the $3$-simplex $\Delta_3$, consist of separable states
of rank four, and so on. We consider the set $F$ of all PPT states whose ranges are contained in the span of $\{|x_i\otimes y_i\rangle\}$
such that the ranges of partial conjugates are contained in the span of $\{|\bar x_i\otimes y_i\rangle\}$. This is a face of the
convex set $\mathbb T$ consisting of all PPT states \cite{ha_kye_04}.
 If the partial conjugates of six product vectors also span the $5$-dimensional space,
then we  see that the face $\Delta_5$ (we use the same notion with the
simplex itself) is inscribed in $F$ so that fifteen faces $\Delta_3$ are located on the boundary of $F$. With a mild additional
condition, we also show that the interior of six maximal faces $\Delta_4$ are contained in the interior of $F$.
Remaining boundary of $F$ consists of PPT entangled edge states of rank four, which are extreme points of $F$.
We also show that every edge state of rank four arises in this way.
From this, we see that
every PPT entangled edge state of rank four is the difference  of scalar multiples of two separable states with rank five.

We have another byproducts with this picture. We see that if we take
two boundary points of $F$ so that they are \lq outside\rq\  of
different maximal faces, then the convex sums of these two points
intersect separable states in $\Delta_5$. This gives an affirmative
answer to the Problem 1 in \cite{chen_dj_ext_PPT} which asks if a
sum of two extreme PPT entangled states may be separable. For a given
entangled state, it is an important problem to find the nearest
separable state. See \cite{pitt}. For generic $3\ot 3$ PPT entangled
states of rank five, our picture
gives a simple method to find the nearest separable state within the
face of PPT states determined by $\varrho$.

In the above discussion, it may happen that the partial conjugates
of six product vectors span the $6$-dimensional space. In this case,
we also have a simplicial face isomorphic to the $5$-simplex
$\Delta_5$. It is interesting to note that the interior points
$\varrho$ of this face have unique decompositions but have
asymmetric ranks: $\rk\varrho\neq\rk\varrho^\Gamma$. It turns out
that the convex combinations of them and the edge states of rank
five constructed in \cite{kye_osaka} give rise to PPT states
$\varrho$ of type $(5,9)$, that is, $\rk\varrho=5$ and
$\rk\varrho^\Gamma=9$. PPT states of such type were missing in a
numerical search \cite{lein}, but the existence has been  noticed in
\cite{hhms}. We give explicit examples of such states.

The {\sl length} of a separable state $\varrho$ is defined by the minimum number $k$ for which the expression (\ref{sep}) is possible.
P. Horodecki \cite{p-horo} showed that the length of an $m\otimes n$ separable state does not exceed $(mn)^2$.
The notion of length was defined in \cite{DiVin}, where it was shown that the length is invariant under the
operation of partial transpose, and so the length may be strictly greater than the rank of the state.
See also \cite{uhlmann,wootters}.
 Our picture also gives concrete examples.
Any separable state $\varrho$ in the interior of $\Delta_5$ has length six,
while   rank of $\varrho$ is five. The rank of $\varrho^\Gamma$ may be five or six.  Therefore,
there exists a state $\varrho$ whose length is greater than the rank of the state
even in the case of $\text{\rm rank}\, \varrho =\text{\rm rank} \, \varrho^{\Gamma}$.
But the lengths of these examples do not exceed $mn$.
Now, it is  natural to ask if there exists an $m\otimes n$ separable state of  length exceeding the number $mn$.
In a recent paper \cite{chen_dj_semialg}, Chen and Djokovi\'c showed that
such a separable state exists when $(m-2)(n-2)>1$, and conjectured \cite[Conjecture 10]{chen_dj_semialg}
that a certain number may be the maximum length. By this
conjecture, the length of $3\ot 3$ separable states do not exceed $9$.

We note that the face $\Delta_5$ of $\mathbb S$ is induced by the face $F$ of $\mathbb T$, in the sense that $\Delta_5$ is the intersection
of $\mathbb S$ and a face $F$ of $\mathbb T$. We call such faces of $\mathbb S$ {\sl induced} ones, as it was introduced in \cite{choi_kye}.
Note that the maximal faces $\Delta_4$ are not induced.
We construct a simplicial face isomorphic to the $9$-simplex $\Delta_9$ with ten extreme points.
This face is turned out to be non-induced.
Our construction gives an example of a $3\otimes 3$ separable state
whose length is $10$, to disprove the conjecture mentioned above.
We note that a face of a face is again a face.
Therefore, we have simplicial faces with arbitrary affine dimensions less than $10$, considering faces of $\Delta_9$,

In the $2\otimes n$ cases, the results in \cite{alfsen,chen_dj_ext_PPT,kirk} give us simplicial faces isomorphic to
the $(n-1)$-simplex. We show for $n\ge 3$ that $n+1$ choices of vectors in
$\mathbb C^2\ot\mathbb C^n$ also give rise to
simplicial faces isomorphic to the $n$-simplex $\Delta_n$,
under suitable linear independence assumption. In this case, we do not have the pictures as in the $3\ot 3$ case,
since there is no PPT entangled states outside of $\Delta_{n}$.

We note that the convex hull of $\{|x_i\otimes y_i\rangle\langle x_i\otimes y_i|:i=1,2,\dots,n\}$ is a simplex
if and only if they are linearly independent in the
{\sl real} vector space of all Hermitian matrices.
It should be noted that this convex hull need not to be a face.
In this context, we begin this paper to search for
conditions on the family $\{|x_i\otimes y_i\rangle\}$ for which
$\{|x_i\otimes y_i\rangle\langle x_i \otimes y_i|\}$ is linearly independent in the next section.
In Section \ref{sec:qutrit}, we consider the $3\ot 3$ case to describe the picture explained above,
and construct separable states with unique decompositions which have asymmetric ranks in Section \ref{SPA}.
We  explain  the notion of
induced faces of $\mathbb S$ in Section \ref{indeced_face}, and construct a $3\ot 3$ separable state which determines the face isomorphic to
the $9$-simplex $\Delta_9$. Finally, we consider the $2\ot n$ cases in Section \ref{sec:qubit-qudit}, and discuss
briefly higher dimensional cases and related problems in the final section.

\section{linear independence of pure product states}\label{sec:indep}

In this section, we investigate to what extent pure product states may be linearly independent.
For this purpose, the notions of {\sl generalized unextendible product basis} and {\sl general position}
are very useful. These two notions play central roles for the description of $3\ot 3$ PPT entangled states of rank four,
as it was developed independently by Skowronek \cite{sko}, and Chen and Djokovi\'c \cite{chen}.
A finite family of product vectors $\{|x_i\otimes y_i\rangle\}$ is said be a generalized unextendible product basis if
there is no product vector in the orthogonal complement of the span of $\{|x_i\otimes y_i\rangle\}$.
It was shown in \cite{sko} that the set $\{|x_i\otimes y_i\rangle \in\mathbb C^m\ot\mathbb C^n:i=1,2,\dots,k\}$
with $k\ge m+n-1$ is a generalized unextendible product basis
if and only if for any partition $\{1,2,\dots,k\}=I\cup J$, at least one of $\{|x_i\rangle:i\in I\}$ or $\{|y_j\rangle :j\in J\}$
spans the whole spaces.

On the other hand,
a family of product vectors
$\{|x_i \otimes y_i\rangle \in \mathbb C^m\otimes\mathbb C^n : i\in I\}$ is said to be in
general position if for any $J\subset I$ with $|J|\le m$ the set $\{|x_j\rangle: j\in J\}$
is linearly independent, and for any $K\subset I$ with $|K|\le n$ the set $\{|y_k\rangle :k\in K\}$ is linearly independent,
where $|J|$ denotes the cardinality of the set $J$.
It is clear that product vectors $\{|x_i\otimes y_i\rangle :i=1,2,\dots,k\}$ with $k\ge m+n-1$
in general position form a generalized unextendible product basis.
It was shown in \cite{sko} that if $k= m+n-1$ then
they are in general position if and only if they form a generalized unextendible product basis.
Therefore, two notions are equivalent for five product vectors in $\mathbb C^3\ot\mathbb C^3$.
For six product vectors in $\mathbb C^3\ot\mathbb C^3$,  there exists a generalized unextendible product basis which is not in general position. See the last paragraph of Section \ref{sec:qutrit} for such a concrete example. Therefore,
these two notions are different in general.

It was shown in Lemma 28 of \cite{chen_dj_ext_PPT} that if
$\{|x_i\otimes y_i\rangle \in \mathbb C^m\otimes \mathbb C^n:i=1,2,\dots,k\}$ with $k\le m+n-2$ is in general position then
they are linearly independent and $|x_i \otimes y_i\rangle$ are the only product vectors in the span of them.
It is clear that if a set of product vectors is linearly independent then the corresponding pure product states are also
linearly independent, but the converse is not true. Therefore, this result tells us that they determines a simplicial face.
In the $3\ot 3$ case, this gives us a simplicial face isomorphic to the $3$-simplex $\Delta_3$, as it was explained in Introduction.
Focusing on the linear independence, we have the following:

\begin{proposition}\label{prop:prod_vec}
Let $1\le k \le m+n-1$. If $k$ product vectors
$$
\{|x_i \otimes y_i\rangle\,:\,i=1,2,\dots, k\}
$$
are in general position in $\mathbb C^m\otimes \mathbb C^n$,
then they are linearly independent.
\end{proposition}
\begin{proof} We may assume that $k>n$. Suppose that  $\sum_{i=1}^k c_i |x_i\rangle \otimes |y_i\rangle =0$.
For any vector $|\omega\rangle$ in the orthogonal complement of the space $\text{span}\{|y_i\rangle : 1\le i \le n-1\}\subset \mathbb C^n$,
we see that
\[
\langle \omega|y_i\rangle=0 \quad (1\le i \le n-1),\qquad \langle \omega|y_j\rangle\neq 0 \quad (n\le j \le k),
\]
since product vectors $|x_i \otimes y_i\rangle$ are in general position. We define a linear map $T_{\omega}$
from $\mathbb C^m\otimes \mathbb C^n$ to $\mathbb C^m$ by $T_{\omega}(|x\rangle \otimes |y\rangle)=\langle \omega|y\rangle x$. Then we have
\[
T_{\omega}\left(\sum_{i=1}^k c_i |x_i\rangle \otimes |y_i\rangle\right)=0
\Longleftrightarrow \sum_{j=n}^k c_j \langle \omega|y_j\rangle |x_j\rangle =0.
\]
Note that $2\le k-n+1 \le m$.
Since $k-n+1$ vectors $|x_j\rangle$ are linearly independent and $\langle \omega|y_j\rangle \neq 0 $
for $n\le j\le k$, we see that all $c_j=0$ for $n\le j\le k$. This implies that $\sum_{i=1}^{n-1} c_i|x_i\rangle \otimes |y_i\rangle=0$,
and thus we see that  $c_i=0$ for $1\le i \le n-1$. Consequently, we get $c_i=0$ for all $1\le i\le k$. This completes the proof.
\end{proof}

We note that a $3\otimes 3$ PPT entangled edge state with rank four has the $5$-dimensional kernel
which has six product vectors in general position by \cite{chen}. Therefore,
the number $m+n-1$ in Proposition \ref{prop:prod_vec} is the best choice  at least for the case of $m=n=3$.
We have seen that product vectors in general position give rise to linearly independent pure product vector states
whenever the number of vectors is less than $m+n$. The following proposition extends this.

\begin{proposition}\label{prop:simplex}
Let $1\le k \le 2m+2n-3$.
If $k$ product vectors
$$
\mathcal S:=\{|z_i\rangle=|x_i\otimes y_i\rangle :i=1,2,\dots,k\}
$$
are in  general position in $\mathbb C^m\otimes \mathbb C^n$,
then the corresponding pure product states are linearly independent.
\end{proposition}
\begin{proof}
From Proposition~\ref{prop:prod_vec}, any $m+n-1$ product vectors of $\mathcal S$ are linearly independent.
So we may assume that $m+n-1<k\le 2m+2n-3$. We consider two subspaces
$$
\begin{aligned}
\mathcal X&= \text{\rm span}\{|x_i\rangle: 1\le i \le m-1\}\subset \mathbb C^m,\\
\mathcal Y&= \text{\rm span}\{|y_i\rangle: m\le i \le n+m-2\}\subset \mathbb C^n.
\end{aligned}
$$
Choose two vectors $|\phi\rangle \in \mathbb C^m$ and $|\psi \rangle \in \mathbb C^n$ such that
$|\phi\rangle \in \mathcal X^{\perp}$ and $|\psi\rangle \in \mathcal Y^{\perp}$.
We note that
\begin{equation}\label{eq:lin_ind}
n+m-1\le i \le k \ \Longrightarrow\ \langle x_i|\phi\rangle \langle y_i|\psi\rangle \neq 0,
\end{equation}
since product vectors of $\mathcal S$ are in general position.

Now, we suppose that $\sum_{i=1}^k c_i|z_i\rangle \langle z_i|=0$. Then we have that
\[
0=\sum_{i=1}^k c_i |x_i\otimes y_i\rangle \langle x_i\otimes y_i| \phi\otimes \psi\rangle
 =\sum_{i=n+m-1}^k c_i \langle x_i|\phi\rangle \langle y_i|\psi\rangle |x_i\otimes y_i\rangle.
\]
Since any $m+n-1$ product vectors of $\mathcal S$ are linearly independent, we see that
$c_i=0$ for all $n+m-1\le i \le k$ from the condition \eqref{eq:lin_ind}.
On the other hand, we know that $n+m-2$ pure product states
$\{|z_i\rangle \langle z_i|:1\le i \le n+m-2\}$ are linearly independent since product vectors
$\{|z_i\rangle :1\le i\le n+m-2\}$ are linearly independent. Consequently,
we see that $c_i=0$ for all $1\le i \le k$. This completes the proof.
\end{proof}

As for the product vectors which form a generalized unextendible product basis, we have the following:

\begin{proposition}\label{3x3}
Suppose that six product vectors $\{|x_i\otimes y_i\rangle:i=1,2,\dots,6\}$ in $\mathbb C^3\otimes\mathbb C^3$
are pairwise distinct, and form a generalized unextendible product basis. Then
the corresponding six pure product states are linearly independent.
\end{proposition}

\begin{proof}
Suppose $\sum_{i=1}^6 c_i |z_i\rangle \langle z_i|=0$.
Consider two sets $\mathcal X=\{|x_1\rangle,|x_2\rangle,|x_3\rangle \}$
and $\mathcal Y=\{|y_4\rangle,|y_5\rangle,|y_6\rangle\}$ from product vectors. By assumption,
at least one of them, say $\mathcal X$, spans the whole space $\mathbb C^3$.
Then, we see the following
\[
\text{\rm dim}(\text{\rm span}(\mathcal Y))\ge 2,\quad \text{\rm dim}(\text{\rm span}(\mathcal Y\cup\{y_i\}))=3,
\]
for each $i=1,2,3$, by assumption. Therefore,
for each $i=1,2,3$, we can choose two vectors $|y_{i_1}\rangle$ and $|y_{i_2}\rangle$ in $\mathcal Y$ such that
$y_i\notin \text{\rm span}\{|y_{i_1}\rangle,|y_{i_2}\rangle\}$.
We denote the element of $\mathcal Y\setminus \{|y_{i_1}\rangle,|y_{i_2}\rangle\}$ by $|y_{i_3}\rangle$.
Now, for each $i=1,2,3$, we choose two vectors $|\phi_i\rangle $ and $|\psi_i\rangle$ such that
\[
|\phi_i\rangle \in \left(\text{\rm span}(\mathcal X\setminus \{|x_i\rangle\})\right)^{\perp},\quad
|\psi_i\rangle \in \left(\text{\rm span}\{|y_{i_1}\rangle,|y_{i_2}\rangle\}\right)^{\perp}.
\]
We note that $\langle x_i|\phi_i\rangle\langle y_i|\psi_i\rangle \neq 0$ because
$\text{\rm dim}(\mathcal X)=3$ and $y_i\notin \text{\rm span}\{|y_{i_1}\rangle,|y_{i_2}\rangle\}$. Therefore, we have that
\[
0=\sum_{k=1}^6 c_k|z_k\rangle \langle z_k|\phi_i\otimes \psi_i\rangle
=c_i \langle x_i|\phi_i\rangle\langle y_i|\psi_i\rangle |x_i\otimes y_i\rangle
+c_{i_3}\langle x_{i_3}|\phi_i\rangle \langle y_{i_3}|\psi_{i_3}\rangle |x_{i_3}\otimes y_{i_3}\rangle
\]
for each $i=1,2,3$. Since $|x_i\otimes y_i\rangle $ is not parallel to $|x_{i_3}\otimes y_{i_3}\rangle$, we can conclude that $c_1=c_2=c_3=0$.

Now,  we  choose a vector $|\psi\rangle \in \mathbb C^3$ and $|y_{k_1}\rangle\in \mathcal Y=\{|y_4\rangle,|y_5\rangle,|y_6\rangle\}$ such that
\[
\langle y_{k_1}|\psi\rangle =0 \ \text{ and }\
\langle y_{k_2}|\psi\rangle \langle y_{k_3}|\psi\rangle \neq 0,
\]
where $\{|y_{k_1}\rangle,|y_{k_2}\rangle,|y_{k_3}\rangle\}=\mathcal Y$. This is possible
because $\text{\rm dim}(\text{\rm span}(\mathcal Y))\ge 2$. Then we choose a vector $|\phi\rangle \in \mathbb C^3$ such that
$\langle x_{k_2} |\phi\rangle \langle x_{k_3}|\phi\rangle \neq 0$.
Consequently, we see that $c_{k_2}=c_{k_3}=0$ from the following equation
\[
0=\sum_{i=4}^6 c_i|z_i\rangle \langle z_i|\phi\otimes \psi\rangle
=c_{k_2} \langle x_{k_2}|\phi\rangle \langle y_{k_2}|\psi\rangle |x_{k_2}\otimes y_{k_2}\rangle
+c_{k_3} \langle x_{k_3}|\phi\rangle \langle y_{k_3}|\psi\rangle |x_{k_3}\otimes y_{k_3}\rangle.
\]
Finally, we see that $c_{k_1}=0$. This completes the proof.
\end{proof}

Very recently, it was shown in \cite{cohen} that if a separable
state $\varrho$ in (\ref{sep}) is not uniquely decomposed then there exists a subset $I\subset \{1,2,\dots,k\}$
with the cardinality $|I|\ge 2$ satisfying
\begin{equation}\label{cohen-inq}
\dim\spa\{|x_i\rangle : i\in I\}+
\dim\spa\{|y_i\rangle : i\in I\}\le |I|+1.
\end{equation}
If $k\le m+n-2$ then this inequality is very useful to infer unique decomposition of a separable state.
Indeed, if $k\le m+n-2$ and $\{|x_i\rangle \ot |y_i\rangle\}$  in (\ref{sep}) is in general position
then every subset $I$ violates this inequality,
and so we get the unique decomposition. This recovers the result in \cite{chen_dj_ext_PPT}.
On the other hand, if $k> m+n-2$ then
the value of the left hand side of (\ref{cohen-inq}) is $m+n$ in most cases,
and the inequality is not violated in these cases.
This means that we need another arguments for unique decompositions of separable states whose ranks exceed $m+n-2$.
The authors are grateful to Scott Cohen for useful discussions on the above result.

\section{Two qutrit separable states with unique decompositions}\label{sec:qutrit}

In this section, we restrict our attention to the two-qutrit case to look for separable states with unique decompositions.
We will concentrate on the rank five separable states, since states of rank up to four are covered by the results in
\cite{alfsen,chen_dj_ext_PPT,kirk}, as it was discussed already. We begin with the following simple observation,
as it was discussed in Introduction.

\begin{proposition}\label{prel}
Suppose that a finite family ${\mathcal P}$ of product vectors in $\mathbb C^m\ot\mathbb C^n$ has the following two
properties:
\begin{enumerate}
\item[(A)]
The corresponding family $\{|z\rangle\langle z|:z\in{\mathcal P}\}$ of pure product states is linearly independent,
\item[(B)]
If a product vector $|x\otimes y\rangle$ belongs to the span of ${\mathcal P}$ then it is parallel to a vector in $\mathcal P$.
\end{enumerate}
Then the convex hull $C_{\mathcal P}$ of $\{|z\rangle\langle z|:z\in{\mathcal P}\}$ is a simplicial face of $\mathbb S$.
\end{proposition}

\begin{proof}
By the condition (A), we see that $C_{\mathcal P}$ is a simplex. To see that it is a face of $\mathbb S$,
suppose that $\varrho_1,\varrho_2\in\mathbb S$ and $\varrho_t:=(1-t)\varrho_0+t\varrho_1$ belongs to $C_{\mathcal P}$.
Then, it is clear that the range space of $\varrho_0:=\sum|w_i\rangle\langle w_i|$ is contained in the range space of $\varrho_t$, which is again
contained in the span of $\mathcal P$. By the condition (B), we have $|w_i\rangle \in\mathcal P$, and so
it follows that $\varrho_0\in C_{\mathcal P}$. Similarily, we can show that $\rho_1\in C_{\mathcal P}$. Therefore, $C_{\mathcal P}$ is a face of $\mathbb S$.
\end{proof}

We note that generic four dimensional subspaces of $\mathbb C^3\otimes\mathbb C^3$ have no product vectors, and their orthogonal complements
have exactly six product vectors
\begin{equation}\label{6-prod}
|z_i\rangle=|x_i\otimes y_i\rangle,\qquad i=1,2,\dots,6,
\end{equation}
which form a generalized unextendible product basis.
By Proposition \ref{3x3}   and Proposition \ref{prel},
we see that a convex combination of product states $|z_i\ran\lan z_i|$ with $i=1,2,\dots,6$ has
a unique decomposition, and determines a face $\Delta_5$ of the convex set $\mathbb S$ of all separable states.
This face is affinely isomorphic to the $5$-simplex with six extreme points. The state
$$
\varrho_0=\frac 16\sum_{i=1}^6 |z_i\ran\lan z_i|
$$
is located at the center of the simplicial face $\Delta_5$.
We note that any interior points of the simplex $\Delta_5$ give rise to examples of separable states whose lengths are strictly greater than
the maximum of the ranks of themselves and partial transposes, if the partial conjugates of \eqref{6-prod} also span the $5$-dimensional space.
Existence of such separable states has been claimed in \cite{chen_dj_semialg} without explicit construction.

From now on throughout this section, we suppose that the partial conjugates of (\ref{6-prod}) also
span the $5$-dimensional space. Note that the kernel of any PPT entangled edge states of rank four have this property.
We  consider a maximal face $\Delta_4$ of
$\Delta_5$ which is the convex hull of five product states
determined by five product vectors among (\ref{6-prod}), say,
$\{|z_1\rangle,\dots, |z_5\rangle\}$. We first consider the case
when   both these five product vectors and their partial conjugates  are linearly independent, and so
they span the   ranges of $\varrho_0$ and $\varrho_0^\Gamma$, respectively.   Any interior point $\sigma$ of
$\Delta_4$ is the convex combination of $\{|z_i\ran\lan
z_i|:i=1,2,\dots,5\}$ with nonzero coefficients, and so the range
spaces of $\varrho_0 $ and $\sigma$ coincide. Therefore, there is
$\epsilon_1>0$ such that $\sigma-\epsilon_1\varrho_0 $ is still
positive. We apply the same argument to the partial conjugates of
$\varrho_0 $ and $\sigma$ to see that there is $\epsilon_2>0$ such
that $\sigma-\epsilon\varrho_0 $  with $\epsilon=\text{min}\{\epsilon_1,\epsilon_2\}$ is a PPT state when normalized. We
note that this argument cannot be applied to a maximal face
$\Delta_3$ of $\Delta_4$ whose interior point has the range space with five linearly independent product vectors.
In this case, the range space of an interior point of $\Delta_3$ is a proper subspace of an interior point of $\Delta_4$.
We note that $\sigma-\epsilon\varrho_0 $ is entangled,
because the range space has only six product vectors in
(\ref{6-prod}) whose convex combination $\Delta_5$ does not contain
$\sigma-\epsilon\varrho_0 $. We take the largest number $\epsilon>0$
so that
\begin{equation}\label{exttt-ppt}
\rho=\sigma-\epsilon\varrho_0
\end{equation}
is of PPT. By the maximality of $\epsilon$, we see that the rank of
$\rho$ is less than $5$. Since every PPT state whose rank is less
than or equal to $3$ is separable by \cite{hlvc}, we see that $\rho$
is of rank four. By the results in \cite{chen,sko}, we conclude that
$\rho$ is a PPT entangled edge state which is an extreme point in
the convex set $\mathbb T$ of all PPT states.

Before going further, we remind the readers of the facial structures  of the convex set $\mathbb T$.
For a pair $(D,E)$ of subspaces of $\mathbb C^m\otimes\mathbb C^n$, it is easy to see that the set
$$
\tau(D,E):=\{\rho\in\mathbb T: {\mathcal R}\rho\subset D,\ {\mathcal R}\rho^\Gamma\subset E\}
$$
is a face of $\mathbb T$, which is possibly empty. It was shown in \cite{ha_kye_04} that every face of $\mathbb T$ is in this form
for a pair $(D,E)$ of subspaces. It was also shown that the interior is given by
$$
\inte \tau(D,E)=\{\rho\in\mathbb T: {\mathcal R}\rho= D,\ {\mathcal R}\rho^\Gamma= E\}.
$$
It is unknown which pairs of subspaces give rise to a nontrivial face of $\mathbb T$, except for the
case of $m=n=2$. See \cite{ha_kye_04}.

In the above situation, let $D_0$ be the span of the product vectors in (\ref{6-prod}) and $E_0$ the span of their partial conjugates.
Then $\varrho_0$ is an interior point of $\tau(D_0,E_0)$.
We note that $\Delta_5= \tau(D_0,E_0)\cap\mathbb S$, and the state $\rho$ given by (\ref{exttt-ppt})
is an extreme point of $\tau(D_0,E_0)$.
We have shown that the interior of the maximal face $\Delta_4$ of $\Delta_5$ is contained in the interior of the face $\tau(D_0,E_0)$.
If $\{|z_1\rangle,\dots, |z_5\rangle\}$ is linearly dependent and spans a proper subspace of ${\mathcal R}\varrho_0$,
then the above process is not possible. If we take an interior point
$\sigma$ of the face $\Delta_4$ then the line segment from $\varrho_0 $ to $\sigma$ cannot be extended within the convex set $\mathbb T$.
This means that the maximal face $\Delta_4$ of $\Delta_5$ lies on the boundary of the face $\tau(D_0,E_0)$ of $\mathbb T$.
We summarize the above discussion in general situations.

\begin{theorem}\label{main}
Suppose that a finite family ${\mathcal P}$ of product vectors in $\mathbb C^m\ot\mathbb C^n$
gives rise to the simplicial face $\mathcal C_{\mathcal P}$ of $\mathbb S$.  Then,
for a subfamily ${\mathcal Q}$ of ${\mathcal P}$, the following are equivalent:
\begin{enumerate}
\item[(i)]
$\spa{\mathcal Q}=\spa{\mathcal P}$   and $\spa{\bar{\mathcal Q}}=\spa{\bar{\mathcal P}}$, where $\bar{\mathcal P}$ denotes the set of
partial conjugates of members in $\mathcal P$,
\item[(ii)]
$\inte C_{\mathcal Q}\subset \inte\tau(\spa{\mathcal P},\spa{\bar{\mathcal P}})$,
\item[(iii)]
A line segment from an interior point of $C_{\mathcal P}$ to an interior point of $C_{\mathcal Q}$ can be extended within the convex set
$\mathbb T$, to get PPT entangled states.
\end{enumerate}
\end{theorem}

In Proposition \ref{prel}, we note that the condition (B) is not necessary
to get simplicial faces,
by considering the face $\Delta_4$ in the above discussion.   In Section \ref{indeced_face},   we consider those cases
in a systematic way. Theorem \ref{main} tells us that $\spa{\mathcal Q}\subsetneqq\spa{\mathcal P}$ if and only if
$C_{\mathcal Q}$ is on the boundary of $\tau(\spa{\mathcal P},\spa{\bar{\mathcal P}})$.

We note that if six product vectors (\ref{6-prod}) are in general position then they satisfy  the condition (A)
by Proposition \ref{prop:simplex}. Furthermore, every five of them are also in general
position, and so satisfy the condition (i) by Proposition \ref{prop:prod_vec}.
Therefore, the simplex $\Delta_5$ is inscribed in the face $\tau(D_0,E_0)$ of $\mathbb T$.
The boundary of $\tau(D_0,E_0)$ consists of
fifteen $3$-simplices from the boundary of $\Delta_5$
and extreme PPT entangled edge states of rank four. We also note that interior points of the $3$-simplices
are separable states of rank four.
Unfortunately, the authors could not determine if every generalized unextendible product basis
consisting of six product vectors in $\mathbb C^3\ot\mathbb C^3$ is in general position or not,
when they span a $5$-dimensional space.

\begin{figure}[h!]
\begin{center}
\includegraphics[scale=0.8]{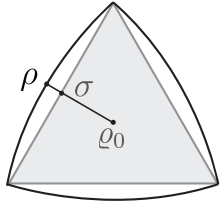}
\end{center}
\caption{The triangle, edges, vertices represent the $5$-simplex, $4$-simplices and $3$-simplices consisting of separable states, respectively.
The triangle is inscribed in the round convex body $\tau(D_0,E_0)$ so that the vertices
are on the boundary of $\tau(D_0,E_0)$. The boundary points
of $\tau(D_0,E_0)$ which is not on the vertices represent PPT entangled edge states of rank four. }
\end{figure}

Now, we take two points $\rho_1,\rho_2$ on the boundary of $\tau(D_0,E_0)$ which are located outside of different maximal faces.
Then it is clear that the line segment between $\rho_1$ and $\rho_2$ touches the simplex $\Delta_5$.
Since $\rho_1$ and $\rho_2$ are extreme PPT states, we see that a convex combination
of two extreme PPT entangled states may be separable. This gives an affirmative answer to Problem 1 of \cite{chen_dj_ext_PPT}.
Actually, we have a stronger result, as we will see in Corollary \ref{prop:sep}.

So far, we have seen that six product vectors in general position which span a $5$-dimensional space give rise to
PPT entangled edge states of rank four by (\ref{exttt-ppt}). We proceed to show that every edge state $\rho$ of rank four arises in this way.
To do this, we first note that
$\ker\rho$ has exactly six product vectors which are in general position by \cite{chen,sko}.
Therefore, we have the following:

\begin{proposition}\label{ortho}
For any $3\otimes 3$ PPT entangled edge state $\rho$ of rank four, there exists a $3\otimes 3$ PPT entangled edge state of rank four
whose range is orthogonal to the range of $\rho$.
\end{proposition}

Now, let $\rho$ be a $3\otimes 3$ PPT entangled edge state of rank four.
By Proposition \ref{ortho}, there exists a PPT edge state $\sigma$ of rank four whose range is contained in $\ker\rho$.
Then $\ker\sigma$ again has exactly six product vectors $\{|x_i\otimes y_i\rangle:i=1,2,\dots 6\}$ in general position
which generate $\ker\sigma$. Note that $\ker \sigma$ contains the range of $\rho$.

We note that the following relation
$$
\text{\rm Tr}[\rho\, \sigma]=\text{\rm Tr}[(\rho\, \sigma)^{\Gamma}]
=\text{\rm Tr}[\rho^{\Gamma} \sigma^{\Gamma}]=\text{\rm Tr}[\sigma^{\Gamma}\rho^{\Gamma}]
$$
holds for PPT states $\rho$ and $\sigma$. We also note that if $\tr(\rho\sigma)=0$ for positive $\rho$ and $\sigma$ then $\rho\sigma=0$. Therefore,
we may conclude that if the range of $\rho$ is contained in $\ker\sigma$ then the range of $\rho^\Gamma$ is also contained in
$\ker\sigma^\Gamma$, as it was already observed in \cite{hlvc}.
We also note that the easy relation
$$
|x\otimes y\rangle \in\ker\sigma\ \Longleftrightarrow |\bar x\otimes y\rangle\in\ker\sigma^\Gamma
$$
holds for PPT states $\sigma$. See \cite{2xn}.
Therefore, we conclude that $\text{\rm Ker}(\sigma^{\Gamma})$
is spanned by six product vectors $\{|\bar x_i\otimes y_i\rangle:1,2,\dots, 6\}$, since they are also in general position.
Now, the face $\tau(\ker\sigma,\ker\sigma^\Gamma)$ of $\mathbb T$ contains the $5$-simplex $\Delta_5$
which is the convex hull of six product vectors $|x_i\otimes y_i\rangle$. We also note that
interior of $\Delta_5$ is contained in the interior of the face $\tau(\ker\sigma,\ker\sigma^\Gamma)$, and
$\rho$ is on the boundary of the face $\tau(\ker\sigma,\ker\sigma^\Gamma)$.
Therefore, we have shown that every PPT entangled edge state arises from six product vectors in general position
spanning $5$-dimensional space, and is of the form (\ref{exttt-ppt}).

\begin{theorem}\label{thm:difference}
Every two-qutrit PPT entangled edge state of rank four is the difference of scalar multiples of two separable states with rank five.
\end{theorem}

We also have the following corollary in the relation with the question in \cite{chen_dj_ext_PPT} mentioned before.

\begin{corollary}\label{prop:sep}
For any $3\otimes 3$ PPT entangled edge state $\rho_1$ of rank four, there exists a PPT entangled edge state $\rho_2$
of rank four such that
parts of the line segment between $\rho_1$ and $\rho_2$ are separable.
\end{corollary}

Now, suppose that $\rho$ is a PPT state of rank five. In generic cases,
the range of $\rho$ has exactly six product vectors $\{|x_i\ot y_i\rangle\}$ which determines a simplicial face $\Delta_5$.
It is clear that $\rho$ is separable if and only if the linear equation
$$
\sum_{i=1}^6 \lambda_i |x_i\ot y_i\rangle\langle x_i\ot y_i|=\rho,\qquad \sum_{i=1}^6\lambda_i=1
$$
has a solution $(\lambda_1,\dots,\lambda_6)$ with $0\le\lambda_i\le 1$.
If this is not the case, we consider the linear equation
$$
\sum_{i=1}^6 \lambda_i |x_i\ot y_i\rangle\langle x_i\ot y_i|
=(1-\mu)\varrho_0+\mu\rho,\qquad \sum_{i=1}^6\lambda_i=1,\qquad 0\le\lambda_i,\mu\le 1
$$
with respect to real unknowns $\lambda_i$ and $\mu$.
If $\varrho$ is expressed as the linear combination of pure product states, then
this equation has infinitely many solutions. But, it has a unique solution
$(\lambda_1,\dots,\lambda_6,\mu)$ under the condition that exactly
one $\lambda_i$ is zero. This solution represents the maximal face
$\Delta_4$ which is the nearest from $\rho$ among six maximal faces
of $\Delta_5$. This suggests a method to find the nearest separable
state to $\rho$, but we could not determine if this is really the
nearest face among all faces of $\mathbb S$.

We illustrate the above discussion with the PPT state  $\rho$ defined by
\[
\rho = \frac 1{3(1+b+1/b)}\begin{pmatrix}
1 & \cdot & \cdot & \cdot & 1 & \cdot & \cdot & \cdot & 1\\
\cdot & b & \cdot & 1 & \cdot & \cdot & \cdot & \cdot & \cdot\\
\cdot & \cdot & 1/b & \cdot & \cdot & \cdot & 1 & \cdot & \cdot\\
\cdot & 1 & \cdot & 1/b & \cdot & \cdot & \cdot & \cdot & \cdot \\
1 & \cdot & \cdot & \cdot & 1 & \cdot & \cdot & \cdot & 1\\
\cdot & \cdot & \cdot & \cdot & \cdot & b & \cdot & 1 & \cdot \\
\cdot & \cdot & 1 & \cdot & \cdot & \cdot & b & \cdot & \cdot\\
\cdot & \cdot & \cdot & \cdot & \cdot & 1 & \cdot & 1/b &\cdot \\
1 & \cdot & \cdot & \cdot & 1 & \cdot & \cdot & \cdot & 1
\end{pmatrix},
\]
 where $b$ is a positive real number with $b\neq 1$, and $\cdot$ denotes 0, as it was constructed in \cite{ha_kye_park}.
Note that $\rho$ is a PPT entangled edge state with
$\rk\rho=\rk\rho^{\Gamma}=4$. If $b=2$, then this is just the first example of
$3\ot 3$ PPT entangled state given by Choi \cite{choi-ppt}.
It was also shown in \cite{ha_kye_park} that $\ker\rho$ has only following six product vectors:
\begin{equation}\label{exam-1}
\begin{aligned}
|x_1\otimes y_1\rangle =&(1,\sqrt b,0)^{\text{\rm t}}\otimes (1,-\frac 1 {\sqrt b},0)^{\text{\rm t}},\, &
& |x_2\otimes y_2\rangle =(1,-\sqrt b,0)^{\text{\rm t}}\otimes (1,\frac 1 {\sqrt b},0)^{\text{\rm t}},\\
|x_3\otimes y_3\rangle =&(0,1,\sqrt b)^{\text{\rm t}}\otimes (0,1,-\frac 1 {\sqrt b})^{\text{\rm t}},\, &
& |x_4\otimes y_4\rangle =(0,1,-\sqrt b)^{\text{\rm t}}\otimes (0,1,\frac 1{\sqrt b})^{\text{\rm t}},\\
|x_5\otimes y_5\rangle =&(\sqrt b,0,1)^{\text{\rm t}}\otimes (-\frac 1 {\sqrt b},0,1)^{\text{\rm t}},\, &
& |x_6\otimes y_6\rangle =(-\sqrt b,0,1)^{\text{\rm t}}\otimes (\frac 1 {\sqrt b},0,1)^{\text{\rm t}}.\\
\end{aligned}
\end{equation}
We denote by $|z_i\rangle$ the normalization of product vector $|x_i\otimes y_i\rangle$.
Then it is also easy to see that
$\mathcal P=\{|z_i\rangle :i=1,2,\dots,6\}$
is in general position, and so it satisfies the conditions (A) and (B) of Proposition \ref{prel}.
Furthermore, any five choice of them satisfies the condition (i) of Theorem \ref{main}.
For each $k=1,2,\dots,6$, we denote by $F_k$ the convex hull of
$\{|z_i\rangle \langle z_i|: |z_i\rangle \in \mathcal P, \ i\neq k\}$.
Then these $F_k$'s are maximal faces
which are isomorphic to $\Delta_4$ with the center
$\rho_k=\frac 15\sum_{\underset{i\neq k}{i=1}}^6 |z_i\rangle \langle z_i|$ of rank five.
We can show that
$$
\text{\rm max}\{\epsilon : \rho_k-\epsilon \varrho_0\in \mathbb T\}=\frac 15,\qquad k=1,2,\dots,6,
$$
where $\varrho_0=\frac1 6\sum_{i=1}^6 |z_i\rangle \langle z_i|$.
Therefore, we see that
$$
\sigma_k=\frac 54 (\rho_k-\frac 15 \varrho_0),\qquad k=1,2,\dots,6
$$
is a PPT entangled edge state of rank four which is the difference
of separable states $\varrho_0$ and $\rho_k$ of rank five.
Note that the range space of $\sigma_k$ is orthogonal to the range space of $\rho$.
We also have
\[
\sigma(t):=t \sigma_i+(1-t)\sigma_j
=\frac {5-6t}{24}|z_i\rangle \langle z_i|+\frac{6t-1}{24}|z_j\rangle \langle z_j|
+\frac{5}{24}\sum_{\underset{k\neq i,\,j}{k=1}}^6|z_k\rangle \langle z_k|.
\]
 By considering the coefficients of $|z_i\rangle \langle z_i|$ and $|z_j\rangle \langle z_j|$,
we see that the state $\sigma(t)$ on the line segment between $\sigma_i$ and $\sigma_j$
is separable if $\frac 16\le t\le \frac 56$. The converse is also ture,
because $\sigma(\frac 16)$ and $\sigma(\frac 56)$ are interior points of maximal faces $F_j$ and $F_i$ of $\mathbb S$, respectively.

We denote by $\{|i\rangle:i=1,2,3\}$ the usual orthonormal basis of $\mathbb C^3$.
If we replace the vector $|x_1\rangle$ in ({\ref{exam-1}) by $|x_1\rangle=|3\rangle$ then it is
straightforward to see that they form a generalized unextendible product vectors. It is cleat that they
are not in general position. It is also easy to see that they span the $6$-dimensional space.
It would be interesting to know whether a generalized unextendible product vectors consisting of six vectors in
$\mathbb C^3\ot\mathbb C^3$ are in general position
when they span the $5$-dimensional space.

\section{Simplicial faces of asymmetric type and optimality of decomposable entanglement witnesses}\label{SPA}

So far, we have considered the five dimensional subspaces which have exactly six product vectors (\ref{6-prod})
whose partial conjugates also span the five dimensional spaces. In this section, we consider the case when the partial
conjugates of (\ref{6-prod}) span the six-dimensional space.
As a byproduct, we construct
PPT states of type $(9,5)$. We say that a PPT state $\varrho$ is of type $(p,q)$ if $\rk\varrho=p$ and $\rk\varrho^\Gamma=q$.

We begin with the $3\ot 3$ PPT entangled edge states $\varrho_{\theta}$ of type $(5,5)$ constructed in \cite{kye_osaka}.
For a fixed positive real number $b$ with $b\neq 1$,   define
$$
\varrho_{\theta}=
\left(
\begin{array}{ccccccccccc}
e^{i\theta}+e^{-i\theta}     &\cdot   &\cdot  &\cdot  &-e^{i\theta}     &\cdot   &\cdot   &\cdot  &-e^{-i\theta}    \\
\cdot   &\frac 1b &\cdot    &-e^{-i\theta}    &\cdot   &\cdot &\cdot &\cdot     &\cdot   \\
\cdot  &\cdot    &b &\cdot &\cdot  &\cdot    &-e^{i\theta}    &\cdot &\cdot  \\
\cdot  &-e^{i\theta}    &\cdot &b &\cdot  &\cdot    &\cdot    &\cdot &\cdot  \\
-e^{-i\theta}     &\cdot   &\cdot  &\cdot  &e^{i\theta}+e^{-i\theta}      &\cdot   &\cdot   &\cdot  &-e^{i\theta}     \\
\cdot   &\cdot &\cdot    &\cdot    &\cdot   &\frac 1b &\cdot &-e^{-i\theta}    &\cdot   \\
\cdot   &\cdot &-e^{-i\theta}    &\cdot    &\cdot   &\cdot &\frac 1b &\cdot    &\cdot   \\
\cdot  &\cdot    &\cdot &\cdot &\cdot  &-e^{i\theta}    &\cdot    &b &\cdot  \\
-e^{i\theta}     &\cdot   &\cdot  &\cdot  &-e^{-i\theta}     &\cdot   &\cdot   &\cdot  &e^{i\theta}+e^{-i\theta}
\end{array}
\right),
$$
where  $-\frac\pi 3<\theta<\frac\pi 3$ and $\theta\neq 0$.
First of all, we note that the kernel of $\varrho_\theta$ is spanned by the following four vectors
\begin{equation}\label{kernel}
\begin{aligned}
|w_1\rangle=&(1,0,0\,;\,0,1,0\,;\,0,0,1)^{\text{\rm t}},\\
|w_2\rangle=&(0,b,0\,;\,e^{i\theta},0,0\,;\,0,0,0)^{\text{\rm t}},\\
|w_3\rangle=&(0,0,0\,;\,0,0,b\,;\,0,e^{i\theta},0)^{\text{\rm t}},\\
|w_4\rangle=&(0,0,e^{i\theta}\,;\,0,0,0\,;\,b,0,0)^{\text{\rm t}}.
\end{aligned}
\end{equation}
From this, it is easy to see that the $5$-dimensional range space of $\varrho_\theta$ has exactly six product vectors, and
the kernel has no product vectors.
We list up all of them with the temporary notation $\omega=\sqrt b \,e^{i\frac{\theta}2}$:
\begin{equation}\label{kernel--}
\begin{aligned}
|z_1\rangle=(1,\sss,0)^{\text{\rm t}}\otimes (\sss,-1,0)^{\text{\rm t}}&=(\sss,-1,0 \,;\, \sss^2,-\sss,0 \,;\, 0,0,0\,)^{\text{\rm t}},\\
|z_2\rangle=(-1,\sss,0)^{\text{\rm t}}\otimes (\sss,1,0)^{\text{\rm t}}&=(-\sss,-1,0 \,;\, \sss^2,\sss,0 \,;\, 0,0,0\,)^{\text{\rm t}},\\
|z_3\rangle=(0,1,\sss)^{\text{\rm t}}\otimes (0,\sss,-1)^{\text{\rm t}}&=(0,0,0\,;\,0,\sss,-1 \,;\, 0,\sss^2,-\sss\,)^{\text{\rm t}},\\
|z_4\rangle=(0,-1,\sss)^{\text{\rm t}}\otimes (0,\sss,1)^{\text{\rm t}}&=(0,0,0\,;\,0,-\sss,-1 \,;\, 0,\sss^2,\sss\,)^{\text{\rm t}},\\
|z_5\rangle=(\sss,0,1)^{\text{\rm t}}\otimes (-1,0,\sss)^{\text{\rm t}}&=(-\sss,0,\sss^2 \,;\, 0,0,0\,;\,-1,0,\sss\,)^{\text{\rm t}},\\
|z_6\rangle=(\sss,0,-1)^{\text{\rm t}}\otimes (1,0,\sss)^{\text{\rm t}}&=(\sss,0,\sss^2 \,;\, 0,0,0\,;\,-1,0,-\sss\,)^{\text{\rm t}}.\\
\end{aligned}
\end{equation}
Define the separable state $\varrho_{\rm sep}$ by
$$
\varrho_{\rm sep}
=\dfrac 1{2b}\sum_{i=1}^6|z_i\rangle\langle z_i|
=
\left(
\begin{array}{ccccccccccc}
2     &\cdot   &\cdot  &\cdot  &-1     &\cdot   &\cdot   &\cdot  &-1    \\
\cdot   &\frac 1b &\cdot    &-e^{-i\theta}    &\cdot   &\cdot &\cdot &\cdot     &\cdot   \\
\cdot  &\cdot    &b &\cdot &\cdot  &\cdot    &-e^{i\theta}    &\cdot &\cdot  \\
\cdot  &-e^{i\theta}    &\cdot &b &\cdot  &\cdot    &\cdot    &\cdot &\cdot  \\
-1     &\cdot   &\cdot  &\cdot  &2     &\cdot   &\cdot   &\cdot  &-1     \\
\cdot   &\cdot &\cdot    &\cdot    &\cdot   &\frac 1b &\cdot &-e^{-i\theta}    &\cdot   \\
\cdot   &\cdot &-e^{-i\theta}    &\cdot    &\cdot   &\cdot &\frac 1b &\cdot    &\cdot   \\
\cdot  &\cdot    &\cdot &\cdot &\cdot  &-e^{i\theta}    &\cdot    &b &\cdot  \\
-1     &\cdot   &\cdot  &\cdot  &-1     &\cdot   &\cdot   &\cdot  &2
\end{array}
\right).
$$
It is clear that $\varrho_{\rm sep}$ is of type $(5,6)$ and has a unique decomposition by Proposition \ref{3x3}
and Proposition \ref{prel}. In this way, we have separable states with unique decompositions which are of asymmetric type.

Now, we consider the PPT state defined by
$$
\varrho=\frac 12(\varrho_{\rm sep}+\varrho_\theta).
$$
It is clear that $\varrho$ is of rank five. It is also easy to see that the partial transpose $\varrho^\Gamma$
$$
\left(
\begin{array}{ccccccccccc}
1+\cos\theta     &\cdot   &\cdot  &\cdot  &-e^{i\theta}     &\cdot   &\cdot   &\cdot  &-e^{-i\theta}    \\
\cdot   &\frac 1b &\cdot    &-\frac{1+e^{-i\theta}}2    &\cdot   &\cdot &\cdot &\cdot     &\cdot   \\
\cdot  &\cdot    &b &\cdot &\cdot  &\cdot    &-\frac{1+e^{i\theta}}2    &\cdot &\cdot  \\
\cdot  &-\frac{1+e^{i\theta}}2    &\cdot &b &\cdot  &\cdot    &\cdot    &\cdot &\cdot  \\
-e^{-i\theta}     &\cdot   &\cdot  &\cdot  &1+\cos\theta    &\cdot   &\cdot   &\cdot  &-e^{i\theta}     \\
\cdot   &\cdot &\cdot    &\cdot    &\cdot   &\frac 1b &\cdot &-\frac{1+e^{-i\theta}}2    &\cdot   \\
\cdot   &\cdot &-\frac{1+e^{-i\theta}}2    &\cdot    &\cdot   &\cdot &\frac 1b &\cdot    &\cdot   \\
\cdot  &\cdot    &\cdot &\cdot &\cdot  &-\frac{1+e^{i\theta}}2    &\cdot    &b &\cdot  \\
-e^{i\theta}     &\cdot   &\cdot  &\cdot  &-e^{-i\theta}     &\cdot   &\cdot   &\cdot  &1+\cos\theta
\end{array}
\right)
$$
is of rank nine. In this way, we have parameterized examples of $3\ot 3$ PPT entangled state $\varrho^\Gamma$ of type $(9,5)$.

Now, we use the duality \cite{eom-kye} between the convex cone
$\mathbb D$ of all decomposable positive linear maps and PPT states,
to consider the dual face of the PPT states $\varrho^\Gamma$ given
by
$$
(\varrho^\Gamma)^\prime
=\{\phi\in\mathbb D:\langle \varrho^\Gamma,\phi\rangle=0\}.
$$
Recall that the bilinear pairing is defined by
$\langle\varrho,\phi\rangle=\tr(\varrho C_\phi^\ttt)$, where
$C_\phi$ is the Choi matrix of the linear map $\phi$. Because
$\varrho^\Gamma$ is of type $(9,5)$, it is clear that the dual face
$(\varrho^\Gamma)^\prime$ consists of completely copositive maps.
Therefore, we conclude that every map in the dual face
$(\varrho^\Gamma)^\prime$ gives rise to a decomposable entanglement
witness. Recall that the kernel of $(\varrho^\Gamma)^\Gamma=\varrho$
is spanned by four vectors in (\ref{kernel}). Therefore,
entanglement witnesses we constructed are the partial transposes of
positive matrices supported on subspaces of $\ker\varrho$.
We consider the entanglement witness $W_\theta$ given by
$$
W_\theta=\left(|w_1\rangle\langle w_1|+\dfrac 1b\sum_{i=2}^4|w_i\rangle\langle w_i|\right)^\Gamma
=\left(
\begin{array}{ccccccccccc}
1   &\cdot   &\cdot  &\cdot  &e^{i\theta}     &\cdot   &\cdot   &\cdot  &e^{-i\theta}    \\
\cdot   &b &\cdot    &1    &\cdot   &\cdot &\cdot &\cdot     &\cdot   \\
\cdot  &\cdot    &\frac 1b &\cdot &\cdot  &\cdot    &1    &\cdot &\cdot  \\
\cdot  &1   &\cdot &\frac 1b &\cdot  &\cdot    &\cdot    &\cdot &\cdot  \\
e^{-i\theta}     &\cdot   &\cdot  &\cdot  &1   &\cdot   &\cdot   &\cdot  &e^{i\theta}     \\
\cdot   &\cdot &\cdot    &\cdot    &\cdot   &b &\cdot &1    &\cdot   \\
\cdot   &\cdot &1    &\cdot    &\cdot   &\cdot &b &\cdot    &\cdot   \\
\cdot  &\cdot    &\cdot &\cdot &\cdot  &1    &\cdot    &\frac 1b &\cdot  \\
e^{i\theta}     &\cdot   &\cdot  &\cdot  &e^{-i\theta}     &\cdot   &\cdot   &\cdot  &1
\end{array}
\right),
$$
which is the partial transpose of a positive matrix supported on $\ker\varrho$.

Recall that an entanglement witness $W$ is called optimal \cite{lew00} if it detects a maximal set of entanglement
with respect to the set inclusion. A decomposable entanglement witness $W$ is the partial transpose
of a positive matrix with the support $E$. Necessary conditions for optimality of a decomposable
witness $W$ were found
in \cite{kye_dec_wit} in terms of $E$: If $W$ is optimal then $E$ satisfies the following:
\begin{enumerate}
\item [(i)]
$E$ has no product vector.
\item [(ii)]
$E^\perp$ has a product vector.
\item [(iii)]
The convex hull of $\{|z\rangle\langle z|^\Gamma: z\in E\}$ is a face of $\mathbb D$ under the correspondence $\phi\mapsto C_\phi$.
\end{enumerate}}

In the $2\otimes n$ cases, it was shown in \cite{aug} that (i) is equivalent to the optimality of $W$. With extra
necessary conditions (ii) and (iii) above, it was shown in \cite{kye_dec_wit} that (i) is not sufficient for optimality in general.
By the discussion above, it is now clear that the witness $W_\theta$ satisfies the above three conditions:
$\ker\varrho$ has no product vector; $(\ker\varrho)^\perp$ is spanned by the six product vectors in (\ref{kernel--});
The convex hull of $\{|z\rangle\langle z|^\Gamma: z\in \ker\varrho\}$ coincides with $(\varrho^\Gamma)^\prime$, and so, it is
actually an exposed face of $\mathbb D$. In the earlier version of this paper, the authors wrote that the witness $W_\theta$ is optimal.
But they found a fault in the reasoning after submission, as it was also observed in \cite{aug-comment}
where it was shown that $W_\theta$ is not optimal. Therefore, we may conclude that
the above three conditions are not sufficient in general for optimality of decomposable entanglement witnesses.

\section{Induced and non-induced faces}\label{indeced_face}

In Section \ref{sec:qutrit}, we have seen two kinds of simplicial faces of the convex set $\mathbb S$. The face $\Delta_5$ arising from
six product vectors in (\ref{exam-1}) is the intersection of $\mathbb S$ with a face of the larger convex set $\mathbb T$
consisting of all PPT states. On the other hands, its maximal face $\Delta_4$ has no such face of $\mathbb T$, because the smallest
face of $\mathbb T$ containing $\Delta_4$ has the intersection $\Delta_5$ with $\mathbb S$.
We distinguish these two cases.

A face $F$ of the convex set $\mathbb S$ is said to be {\sl induced} by a face $\tau(D,E)$ of $\mathbb T$, or
$\tau(D,E)$ {\sl induces} the face $F$ of $\mathbb S$ if the relation
$$
F=\tau(D,E)\cap\mathbb T,\qquad \inte F\subset\inte\tau(D,E)
$$
holds for a pair $(D,E)$ of subspaces in $\mathbb C^m\ot\mathbb
C^n$, as it was introduced in \cite{choi_kye}. This pair is uniquely
determined by the above relation. If a face $F$ of $\mathbb S$ is
induced by a face of $\mathbb T$, then we just say that it is an
{\sl induced} face. A face $\tau(D,E)$ induces a face of $\mathbb S$
if and only if it has a separable state in its interior. It was
shown in \cite{choi_kye} that this is the case if and only if the
pair $(D,E)$ satisfies the range criterion \cite{p-horo}, that is, there exist
product vectors $|x_i\otimes y_i\rangle$ such that
$D=\spa\{|x_i\otimes y_i\rangle\}$ and $E=\spa\{|\bar x_i\otimes
y_i\rangle\}$. This tells us how the partial converse of the range
criterion works.
In order to test for a separable state $\varrho$ to determine an induced face,
we denote by $P_1(\varrho)$ the set of all product vectors $|z\rangle$
such that the state $|z\rangle\langle z|$ belongs to the face determined by $\varrho$, and $P_2(\varrho)$ the set of all product vectors
$|x \otimes y\rangle\in {\mathcal R}(\varrho)$ such that $|\bar x\otimes y\rangle\in {\mathcal R}(\varrho^\Gamma).$
It was shown in \cite{ha_kye_12_osid} that $P_1(\varrho)\subset P_2(\varrho)$ holds in general, and the equality holds if and only
if $\varrho$ determines an induced face of $\mathbb S$.

\begin{proposition}\label{prop:induced_face}
Let $\mathcal P=\{|x_i\otimes y_i\rangle :1,2,\dots,n\}$ be a family of normalized product vectors.
Then the convex hull $C_{\mathcal P}$ of $\{|z\rangle\langle z|:z\in{\mathcal P}\}$
is an induced simplicial face of $\mathbb S$ if and only if
$\mathcal P$ satisfies the condition {\rm (A)} in Proposition {\rm \ref{prel}} together with the following:
\begin{enumerate}
\item[(C)]
If a product vector $z=|x\otimes y\rangle$ belongs to the span of ${\mathcal P}$ and
$|\bar x\otimes y\rangle$ belongs to the span of $\bar{\mathcal P}$ then it is parallel to a vector in $\mathcal P$.
\end{enumerate}
\end{proposition}

\begin{proof}
Suppose that $C_{\mathcal P}$ is an induced simplicial face. Then we have the condition (A) together with the relation
\begin{equation}\label{jjjj}
C_{\mathcal P}=\mathbb S\cap\tau(\spa{\mathcal P},\spa\bar{\mathcal P}).
\end{equation}
To prove the condition (C), suppose that a product vector $|z\rangle$ satisfies the assumption of (C). Then
$|z\rangle\langle z|\in\tau(\spa{\mathcal P},\spa\bar{\mathcal P})\subset C_{\mathcal P}$ by (\ref{jjjj}).
Since $C_{\mathcal P}$ is a simplex,
we see that the condition (C) holds.

For the converse, suppose that the conditions (A) and (C) holds. If $|z\rangle=|x\otimes y\rangle$ is a product vector and
$\varrho=|z\rangle\langle z|$ belongs to $\tau(\spa{\mathcal P},\spa\bar{\mathcal P})$ then we have
$$
|x\otimes y\rangle \in {\mathcal R}\varrho\subset \spa{\mathcal P},\qquad
|\bar x\otimes y\rangle\in {\mathcal R}\varrho^\Gamma\subset \spa\bar{\mathcal P}.
$$
By the condition (C), we see that $\varrho\in C_{\mathcal P}$. This shows one direction of the inclusion of (\ref{jjjj}).
Since the other direction is true always, we have the relation (\ref{jjjj}) which tells us that $C_{\mathcal P}$ is an induced face.
Finally, $C_{\mathcal P}$ is a simplex by the condition (A).
\end{proof}

In the remaining of this section, we construct a simplicial face which is not induced.
We modify the example (\ref{exam-1}) as follows:
\begin{equation}\label{exam-2}
\begin{aligned}
&(1,\sqrt b,0)^{\text{\rm t}} \otimes (1,\frac1{\sqrt b},0)^{\text{\rm t}},\,  &
&(1,-\sqrt b,0)^{\text{\rm t}} \otimes (1,-\frac1 {\sqrt b},0)^{\text{\rm t}},\\
&(0,1,\sqrt b)^{\text{\rm t}} \otimes (0,1,\frac1{\sqrt b})^{\text{\rm t}},\,  &
&(0,1,-\sqrt b)^{\text{\rm t}} \otimes (0,1,-\frac1 {\sqrt b})^{\text{\rm t}},\\
&(\sqrt b,0,1)^{\text{\rm t}} \otimes (\frac1{\sqrt b},0,1)^{\text{\rm t}},\, &
&(-\sqrt b,0,1)^{\text{\rm t}} \otimes (-\frac1 {\sqrt b},0,1)^{\text{\rm t}},
\end{aligned}
\end{equation}
where we still retain the condition $b\neq 1$.
We see that these six product vectors span the $6$-dimensional space whose orthogonal complement is spanned by
\begin{equation}\label{exam-2-perp}
\begin{aligned}
\sqrt{b}|1\rangle\otimes |2\rangle &- \frac 1{\sqrt b} |2\rangle\otimes |1\rangle,\\
\sqrt{b}|2\rangle\otimes |3\rangle &- \frac 1{\sqrt b} |3\rangle\otimes |2\rangle,\\
\sqrt{b}|3\rangle\otimes |1\rangle &- \frac 1{\sqrt b} |1\rangle\otimes |3\rangle.
\end{aligned}
\end{equation}
From this, it is easy to see that there are infinitely many product vectors in the span of product vectors (\ref{exam-2}).
Since all
 entries are real, these vectors do not satisfy the condition (C). Nevertheless, we show that these vectors
make a simplicial face. So we will get a simplicial face which is not induced. Actually we show that this is a subface of a simplicial face isomorphic to the $9$-simplex.

To do this, we consider
the product vectors:
\begin{equation}\label{exam-3}
\begin{aligned}
&(1,1,1)^{\text{\rm t}} \otimes (1,1,1)^{\text{\rm t}},\, &
&(1,1,-1)^{\text{\rm t}} \otimes (1,1,-1)^{\text{\rm t}},\\
&(1,-1,1)^{\text{\rm t}} \otimes (1,-1,1)^{\text{\rm t}},\,  &
&(-1,1,1)^{\text{\rm t}} \otimes (-1,1,1)^{\text{\rm t}}.\\
\end{aligned}
\end{equation}
These four product vectors span the $4$-dimensional space whose orthogonal complement is spanned by
$$
\begin{aligned}
&|1\rangle\otimes |2\rangle - |2\rangle\otimes |1\rangle,\qquad
|2\rangle\otimes |3\rangle - |3\rangle\otimes |2\rangle,\qquad
|3\rangle\otimes |1\rangle - |1\rangle\otimes |3\rangle,\\
&|1\rangle\otimes |1\rangle - |2\rangle\otimes |2\rangle,\qquad
|2\rangle\otimes |2\rangle - |3\rangle\otimes |3\rangle.
\end{aligned}
$$
From this, it is immediate to see that these product vectors satisfy the condition (A) and (B), and so determine
a simplicial face. We denote by $\mathcal P$ the family of normalizations of the above product vectors (\ref{exam-3})
together with six product vectors in
(\ref{exam-2}). We note that both ${\mathcal P}$ and $\bar{\mathcal P}$ span the full space $\mathbb C^3\ot\mathbb C^3$, and so
it is clear that the condition (C) does not hold.   In addition to this fact, we show that $C_{\mathcal P}$ is a simplicial face.  So $C_{\mathcal P}$ is also a simplicial face which is not induced.
This will prove that every interior point of $C_{\mathcal P}$ gives rise to a separable state of rank nine whose length is ten.
This disprove a conjecture in \cite{chen_dj_semialg}, as it was discussed in Introduction.

We first show that the convex set $C_{\mathcal P}$ is a simplex.
To do this, we label the  normalized vectors in (\ref{exam-2}) by $|z_1\rangle,|z_2\rangle,\dots, |z_6\rangle$, and  normalized vectors
in (\ref{exam-3}) by $|z_7\rangle,\dots, |z_{10}\rangle$. Suppose that $\sum_{i=1}^{10}\alpha_i |z_i\rangle\langle z_i|=0$.
If $|w\rangle$ is any one of (\ref{exam-2-perp}) then we have
$\langle z_i|w\rangle \neq 0$ for $i=7,8,9,10$. Since  $\{|z_7\rangle,\dots, |z_{10}\rangle\}$
is linearly independent, we see from the relation
$$
0= \sum_{i=1}^{10}\alpha_i |z_i\rangle\langle z_i|w\rangle =\sum_{i=7}^{10}\alpha_i\langle z_i|w\rangle | z_i\rangle
$$
that $\alpha_i=0$ for $i=7,8,9,10$. Since $\{|z_1\rangle,|z_2\rangle,\dots, |z_6\rangle\}$ is also linearly independent,
we  conclude that $\alpha_i=0$ for $i=1,2,\dots,6$.
Therefore, we see that $\{|z_i\rangle\langle z_i|: i=1,2,\dots,10\}$ is linearly independent, and so
$C_{\mathcal P}$ is a simplex.

In order to show that $C_{\mathcal P}$ is a face, we use the duality between positive linear maps and separable states to
see that $C_{\mathcal P}$ is a dual face of a positive linear map.
We recall the generalized Choi map $\Phi[\alpha,\beta,\gamma]$ between $M_3$ defined by
\[
\Phi[\alpha,\beta,\gamma](X)=
\begin{pmatrix}
\alpha x_{11}+\beta x_{22}+\gamma x_{33} & -x_{12} & -x_{13}\\
-x_{21} & \gamma x_{11}+\alpha x_{22}+\beta x_{33} & -x_{23}\\
-x_{31} & -x_{32} & \beta x_{11}+\gamma x_{22} +\alpha x_{33}
\end{pmatrix}
\]
for $X\in M_3$ and nonnegative real numbers $\alpha,\beta$ and $\gamma$, as it was introduced in \cite{ckl}.
We restrict our attention to the  maps $\Phi[\alpha,\beta,\gamma]$ satisfying the following condition
\[
0< \alpha <1,\quad \alpha+\beta+\gamma=2,\quad \beta \gamma=(1-\alpha)^2,
\]
which implies that $\Phi[\alpha,\beta,\gamma]$ is a positive linear map.
This condition is parameterized by
\[
\alpha (s)=\frac{(1-s)^2}{1-s+s^2},\quad \beta (s)=\frac{s^2}{1-s+s^2},\quad \gamma (s)=\frac 1{1-s+s^2},
\]
with $0<s<\infty$, $s\neq 1$.
Therefore, we obtain parameterized positive maps
\[
\Phi(s)=\Phi\left [\alpha (s),\beta (s),\gamma (s)\right],\qquad 0<s<\infty,\ s\neq 1
\]
as in \cite{ha_kye_11_osid} and \cite{ha_kye_12_osid} (See page 333 in \cite{ha_kye_11_osid} and page 6 in \cite{ha_kye_12_osid}). We also recall that every positive linear map $\Phi(s)$ gives rise to the dual face
\[
\Phi(s)^\prime=\{\rho\in \mathbb S:\text{\rm Tr}(C_{\Phi(s)}^{\text{\rm t}}\,\rho)=0\}
\]
of $\mathbb S$, where $C_{\Phi(s)}^{\text{\rm t}}$ is the transpose of the Choi matrix of the map $\Phi(s)$ given by
\[
\begin{aligned}
C_{\Phi(s)}=&\sum_{i,j=1}^3 |i\rangle \langle j|\otimes \Phi_s(|i\rangle \langle j|)\\
=&\begin{pmatrix}
\alpha (s) & \cdot & \cdot & \cdot & -1 & \cdot & \cdot & \cdot & -1\\
\cdot & \gamma (s) & \cdot& \cdot & \cdot & \cdot & \cdot & \cdot & \cdot\\
\cdot & \cdot &\beta (s) &  \cdot & \cdot & \cdot & \cdot & \cdot & \cdot\\
\cdot & \cdot & \cdot & \beta (s) & \cdot & \cdot & \cdot & \cdot & \cdot\\
-1 & \cdot & \cdot & \cdot & \alpha (s) & \cdot & \cdot & \cdot & -1\\
\cdot & \cdot & \cdot& \cdot & \cdot & \gamma (s) & \cdot & \cdot & \cdot\\
\cdot & \cdot & \cdot& \cdot & \cdot & \cdot & \gamma (s) & \cdot & \cdot\\
\cdot & \cdot & \cdot& \cdot & \cdot & \cdot & \cdot & \beta (s) & \cdot\\
-1 & \cdot & \cdot& \cdot & -1 & \cdot & \cdot & \cdot & \alpha (s)\\
\end{pmatrix}\in M_3\otimes M_3.
\end{aligned}
\]

 It is known from  Theorem 3 in \cite{ha_kye_12_osid} that every product vector whose product state
 belongs to the dual face $\Phi(1/b)'$ is one of the following forms
\begin{equation}\label{complex_prod_vec}
\begin{gathered}
(\bar a_1,\bar a_2,\bar a_3)^{\text{\rm t}}\otimes (a_1,a_2,a_3)^{\text{\rm t}}\quad \text{\rm with }\ |a_1|=|a_2|=|a_3|,\\
(0,\bar a_2,b\bar a_3)^{\text{\rm t}}\otimes (0,a_2,a_3)^{\text{\rm t}}\quad \text{\rm with }\ |a_2|^2=b|a_3|^2,\\
(b\bar a_1,0,\bar a_3)^{\text{\rm t}}\otimes, (a_1,0,a_3)^{\text{\rm t}}\quad \text{\rm with }\ |a_3|^2=b|a_1|^2,\\
(\bar a_1,b\bar a_2,0)^{\text{\rm t}}\otimes (a_1,a_2,0)^{\text{\rm t}}\quad \text{\rm with }\ |a_1|^2=b|a_2|^2,
\end{gathered}
\end{equation}
and so the dual face  $\Phi(1/b)^\prime$ is the convex hull of pure product states corresponding to these
product vectors.
We note that $\mathcal P$ is contained in the dual face $\Phi(1/b)^\prime$.
Furthermore, we see that $\mathcal P$ coincides with the set of product vectors in (\ref{complex_prod_vec})
with real components up to scalar multiplications.
Now, we consider the dual face $(\Phi(s) \circ \text{\rm t})'$ where $\Phi_s\circ \text{\rm t}$ is the composition of
$\Phi_s$ and the transpose map $\text{\rm t}$.
We note that
\[
C_{\Phi(s)\circ \text{\rm t}}
=\sum_{i,j=1}^3 |i\rangle \langle j| \otimes (\Phi(s)\circ \text{\rm t}) (|i\rangle \langle j|)
=\sum_{i,j=1}^3 |i\rangle \langle j|\otimes \Phi(s) (|j\rangle \langle i|)=C_{\Phi(s)}^{\Gamma}.
\]
Thus we see that a product vector $|x\otimes y\rangle $ belongs to the dual face
$(\Phi(s)\circ \text{\rm t})'$ if and only if $|\bar x\otimes y\rangle $ belongs
to the dual face $\Phi(s)'$. From this observation, we easily see that
$|x\otimes y\rangle\in {\mathcal P}$ if and only if
$|x\otimes y\rangle\in \Phi(1/b)'\cap (\Phi(1/b)\circ \text{\rm t})'$, which coincides with the dual face of the map
$\Phi(1/b)+\Phi(1/b)\circ \text{\rm t}$. This shows that the convex set $C_{\mathcal P}$ is just the dual face of the
positive map $\Phi(1/b)+\Phi(1/b)\circ \text{\rm t}$.
Consequently, $C_{\mathcal P}$ is a non-induced simplicial face of $\mathbb S$ with the center
$$
\varrho_0=\frac 1{10} \sum_{i=1}^{10}|z_i\rangle \langle z_i|,
$$
which is  isomorphic the $9$-simplex $\Delta_9$.

We note that the method to get PPT entanglement with Theorem \ref{main}
works in more general situations.
Suppose that $F$ is a face of the convex set $\mathbb S$, which is not necessarily isomorphic to a simplex.
Take an interior point $\varrho_0$ and a boundary point of $\sigma$ of $F$. If
the range spaces of $\varrho_0$ and $\sigma$ coincide, then we can take $\epsilon>0$ such that $\rho=\sigma-\epsilon \varrho_0$
is positive. We note that this state $\rho$ must be entangled. To see this, assume that $\rho$ is separable. Then
we see that $\sigma$ is the convex sum of two separable states $\varrho_0$ and $\rho$, and so
we conclude that $\rho\in F$ by the definition of face.
This contradiction shows that $\rho$ must be entangled. If the range spaces of $\varrho_0^\Gamma$ and $\sigma^\Gamma$ also coincide
then we can take $\epsilon>0$ such that both $\rho=\sigma-\epsilon \varrho_0$ and
$\rho^\Gamma=\sigma^\Gamma-\epsilon \varrho_0^\Gamma$ are positive. Then
we see that $\rho$ must be PPT entangled state by the same argument.

\begin{proposition}
Let $F$ be a face of the convex set $\mathbb S$ with an interior point $\varrho_0$ and a boundary point $\sigma$.
If $\rho=\sigma-\epsilon \varrho_0$ is positive for $\e>0$ then $\rho$ is an entangled state. If both $\rho$ and $\rho^\Gamma$ are positive
then $\rho$ is a PPT entangled state.
\end{proposition}

Now, we return to the above simplicial face $C_{\mathcal P}=\Delta_9$.
For each $k=1,2,\dots,10$, we denote by $F_k$ the convex hull of
$\{|z_i\rangle\langle z_i|:|z_i\rangle \in \mathcal P,\,i\neq k\}$.
Then these $F_k$'s are maximal faces of $C_{\mathcal P}$,
which are isomorphic to $\Delta_8$ with the center
$$
\sigma_k=\frac 19\sum_{\underset{i\neq k}{i=1}}^{10} |z_i\rangle \langle z_i|
$$
of rank nine. We also see that for any interior point $\sigma$ of $F_k$, the range spaces of $\varrho_0$ and $\sigma$ coincide,
and so we can apply the above argument.
Using a symbolic computation program,
we can show that
\[
\lambda_k:=\text{\rm max}\{\epsilon:\sigma_k-\epsilon \varrho_0\in \mathbb T\}=
\begin{cases}
\dfrac{5(1+b)^2}{27(1+8b+b^2)},&\qquad  k=1,2,3,4,5,6,\\
\dfrac{5b}{3(1+8b+b^2)},&\qquad  k=7,8,9,10.
\end{cases}
\]
In this way, we have PPT entangled states $\rho_k=\sigma_k-\lambda_k \rho_0$ with
$\rk\rho_k=\rk\rho_k^{\Gamma}=8$, which are located at the boundary of $\mathbb T$.
\section{qubit-qudit cases}\label{sec:qubit-qudit}

In this section, we turn our attention to the $2\otimes n$ cases, and look for simplicial faces determined by
separable states with rank $n+1$, because the results in \cite{alfsen,chen_dj_ext_PPT,kirk} covers separable states with rank up to $n$.
We take $n+1$ product vectors in general
position. Then these vectors are linearly independent by Proposition
\ref{prop:prod_vec}, and so the corresponding pure product states
satisfy the condition (A) in Proposition~\ref{prel}. We will show
that these vectors also satisfy the condition (C) in Proposition
\ref{prop:induced_face} under a mild assumption. Consequently, the corresponding pure
product states determine an induced simplicial face as follows.

\begin{theorem}\label{2xn_thm}
Suppose that
$\mathcal P=\{ |e_i\otimes f_i\rangle :i=1,2,\dots,n+1\}$ is a family of normalized product vectors,
which are in general position in $\mathbb C^2\otimes \mathbb C^n$. We write
\[
|e_i\rangle=a_{i1}|e_1\rangle +a_{i2} |e_2\rangle\in\mathbb C^2,\quad i=3,4,\dots,n+1,
\]
then we have the following:
\begin{enumerate}
\item[(i)]
If $n=2$ then the face determined by ${\mathcal P}$ has infinitely many extreme points.
\item[(ii)]
Let $n\ge 3$. In this case, if $a_{n+1,1}\,a_{k2}\, \bar a_{n+1,2}\,\bar a_{k1}\notin \mathbb R$ for each $k=3,4,\dots, n$, then
$\mathcal P$ determines an induced simplicial face isomorphic to the $n$-simplex.
\end{enumerate}
\end{theorem}

\begin{proof}
Write $|f_{n+1}\rangle=\sum_{i=1}^n b_i|f_i\rangle$. Then none of ${a_{ij}}'s$ and ${b_k}'s$ is zero by assumption of general position.
Suppose that
a product vector $|x\otimes y\rangle$ satisfies the assumption of condition (C). We also write
$|x\rangle=x_1|e_1\rangle+x_2|e_2\rangle$ and $|y\rangle=\sum_{i=1}^n y_i|f_i\rangle$, with $x_i,y_j\in\mathbb C$.
Then there exist complex numbers $s_i,t_i$ such that
\begin{equation}\label{eq:2xn}
\begin{aligned}
|x\otimes y\rangle&=\sum_{i=1}^{n+1}s_i |e_i\otimes f_i\rangle
=s_1|e_1\otimes f_1\rangle+s_2|e_2\otimes f_2\rangle +\sum_{i=3}^n \sum_{j=1}^2s_i a_{ij}|e_{j}\otimes f_i\rangle \\
&\hskip4truecm +\sum_{j=1}^2\sum_{k=1}^n s_{n+1}a_{n+1,j}b_k|e_j\otimes f_k\rangle,\\
|\bar x\otimes y\rangle&=\sum_{i=1}^{n+1}t_i |\bar e_i\otimes f_i \rangle
=t_1|\bar e_1\otimes f_1\rangle+t_2|\bar e_2\otimes f_2\rangle +\sum_{i=3}^n \sum_{j=1}^2t_i \bar a_{ij}|\bar e_{j}\otimes f_i\rangle \\
&\hskip4truecm +\sum_{j=1}^2\sum_{k=1}^n t_{n+1}\bar a_{n+1,j}b_k|\bar e_j\otimes f_k\rangle .
\end{aligned}
\end{equation}
Since $\{|e_i\otimes f_j\rangle :i=1,2,\ j=1,2,\dots,n\}$ form a basis of $\mathbb C^2\otimes \mathbb C^n$,
we can solve $x_i,y_j$ and $s_k,t_k$
by comparing the coefficients of basis vectors in \eqref{eq:2xn},

If $s_{n+1}=0$, then we see that $x_1y_2=0$ and $x_2 y_1=0$.
It is easy to that
\[
|x\otimes y\rangle =\begin{cases}
s_1|e_1\otimes f_1\rangle &\quad \text{if }x_2=y_2=0,\\
s_2|e_2\otimes f_2\rangle &\quad \text{if }x_1=y_1=0,\\
s_k|e_k\otimes f_k\rangle\ (k=3,4,\dots,n) &\quad \text{if }y_1=y_2=0,
\end{cases}
\]
because product vectors $|e_i\otimes f_i\rangle$ are in general position.
This shows that there exist exactly $n$ product vectors in the span
of ${\mathcal P\setminus \{|e_{n+1}\otimes f_{n+1}\rangle\}}$ whose partial conjugates
belong to the span of partial conjugates of product vectors in ${\mathcal P}$.

Now, we consider the case $s_{n+1}\neq 0$. We look at the coefficients of $|e_1\otimes f_2\rangle$ and $|e_2\otimes f_1\rangle $ to get
\[
\begin{aligned}
x_1y_2&=s_{n+1}a_{n+1,1}\,b_2,\quad
x_2y_1&=s_{n+1}a_{n+1,2}\,b_1,\\
\bar x_1y_2&=t_{n+1}\bar a_{n+1,1}\,b_2,\quad
\bar x_2 y_1&=t_{n+1}\bar a_{n+1,2}\,b_1.
\end{aligned}
\]
We see that $x_i, y_j$ are nonzero, and so we may assume that $x_1=a_{n+1,1}$.
Then we have $y_2=s_{n+1}b_2$ and $t_{n+1}=s_{n+1}$ from the first and third equalities. From the second and forth equalities,
we also have $x_2=r \,a_{n+1, 2}$ and $y_1=s_{n+1}b_1/r$ for a nonzero real number $r$. Therefore, we have
\begin{equation}\label{eq:sol1}
x_1=a_{n+1,1},\quad x_2=ra_{n+1,2},\quad y_1=s_{n+1}\frac{b_1} r,\quad y_2=s_{n+1}b_2,\quad t_{n+1}=s_{n+1}
\end{equation}
for a nonzero real number $r$. If $n=2$ then this shows that there are infinitely many product vectors in the span
of ${\mathcal P}$ whose partial conjugates belong to the span of partial conjugates of product vectors in ${\mathcal P}$,
and this shows that the face determined by ${\mathcal P}$ has infinitely many extreme points. Consequently, we proved the assertion (i).

Next, suppose that $n\ge 3$.
For $k\ge 3$ we look at the  coefficients of $|e_1\otimes f_k\rangle $ and
$|e_2\otimes f_k\rangle$ to solve $y_k,\,s_k,\,t_k$ and $s_{n+1}$.
After substituting $x_1,x_2$ by the values in \eqref{eq:sol1} with $s_{n+1}=t_{n+1}$, we have the following equation
\begin{equation}\label{eq:sol2}
\begin{pmatrix}
-a_{n+1,1}   & a_{k1}  & 0    & a_{n+1,1}b_k\\
-r a_{n+1,2} & a_{k2}  & 0    & a_{n+1,2}b_k\\
-\bar a_{n+1,1}  & 0   & \bar a_{k1} &\bar a_{n+1,1}b_k\\
-r \bar a_{n+1,2}& 0   & \bar a_{k2} &\bar a_{n+1,2}b_k
\end{pmatrix}
\begin{pmatrix}
y_k \\ s_k \\ t_k\\ s_{n+1}
\end{pmatrix}
=\begin{pmatrix} 0\\0\\0\\0\end{pmatrix},
\end{equation}
where the determinant of the $4\times 4$ matrix is given by
\[
(r-1)b_k (a_{n+1,1}\,a_{k2}\, \bar a_{n+1,2}\,\bar a_{k1}-\bar a_{n+1,1}\, \bar a_{k2}\, a_{n+1,2}\, a_{k1}).
\]
This determinant must be zero, because the homogeneous equation (\ref{eq:sol2}) has a nonzero solution by the current
assumption $s_{n+1}\neq 0$.
By the assumption of the statement (ii), we have $r=1$. In this case, the matrix in (\ref{eq:sol2}) is of rank three,
with the solution $s_k=t_k=0$ and $y_k=s_{n+1}b_k$. Therefore, we have
\[
|x\otimes y\rangle =s_{n+1}|e_{n+1}\otimes f_{n+1}\rangle.
\]
By combining the results for the cases of both $s_{n+1}=0$ and $s_{n+1}\neq 0$, $\mathcal P$ satisfies the condition (C) in Proposition~\ref{prop:induced_face}.
We already know that $\mathcal P$ also satisfies the condition (A) in Proposition~\ref{prel}.
Therefore, $\mathcal P$ determines an induced simplicial face by Proposition~\ref{prop:induced_face}.
\end{proof}

It should be noted that every $(n+1)$-dimensional subspace of $\mathbb C^2\ot\mathbb C^n$ has infinitely many product vectors,
as it was shown in Lemma 10 of \cite{2xn}. If every vector in $\mathcal P$ has real entries, then ${\mathcal P}$ and $\bar{\mathcal P}$
coincide and so there are infinitely many product vectors in $\spa{\mathcal P}=\spa{\bar{\mathcal P}}$.
The curious condition $a_{n+1,1}\,a_{k2}\, \bar a_{n+1,2}\,\bar a_{k1}\notin \mathbb R$
may reflect this fact.

In the low dimensional cases, we can now classify all simplicial faces.
In the case of $2\ot n$ with $n=2,3$, two notions of separability and PPT coincide, and so
faces of the convex set $\mathbb S$ are classified in terms of pairs $(D,E)$ of subspaces of $\mathbb C^2\ot\mathbb C^n$.
We first find a necessary condition for a separable state to determine an induced simplicial face. We note that every face is induced in
the $2\ot 2$ and $2\ot 3$ cases.

\begin{proposition}\label{cor-2xn--}
Suppose that a $2\ot n$ separable state $\varrho$ of the form {\rm (\ref{sep})} determines an induced
simplicial face of the convex set $\mathbb S$. Then
$\{|x_i\rangle\}$ is pairwise distinct up to scalar multiplications. In the cases of $n=2,3$,
${\mathcal P}=\{|x_i\otimes y_i\rangle:i=1,2,\dots,k\}$ is in general position.
\end{proposition}

\begin{proof}
Assume that two of $\{|x_i\rangle\}$, say $|x_1\rangle$ and $|x_2\rangle$, are parallel to each other.
For each $|y\rangle \in\spa\{|y_1\rangle,|y_2\rangle\}$, we consider the
product state $\varrho_y=|x_1\ot y\rangle\langle x_1\ot y|$. Then we see that $|x_1\otimes y\rangle\in {\mathcal R}\varrho$
and $|\bar x_1\otimes y\rangle\in{\mathcal R}\varrho^\Gamma$. Therefore, we see that $\varrho_y$ belongs to the face
$\mathbb S\cap\tau({\mathcal R}\varrho,{\mathcal R}\varrho^\Gamma)$ of $\mathbb S$ determined by $\varrho$.
This face is not a simplex because there are infinitely many extreme points $\varrho_y$ for arbitrary
$|y\rangle \in \text{\rm span}\{|y_1\rangle,|y_2\rangle\}$.
Therefore, we see that $\{|x_i\rangle\}$ is pairwise distinct up to scalar multiplications.

For the second claim, it suffices to consider the the case of $n=3$.
If two of $\{|y_i\rangle\}$ are parallel to each other then $\varrho$ cannot determine a simplicial face
by the same reason as above. If three of them, say
$|y_1\rangle, |y_2\rangle, |y_3\rangle$, span a $2$-dimensional space $V$ of $\mathbb C^3$,
then ${\mathcal P}_0=\{|x_i\otimes y_i\rangle:i=1,2,3\}$ is contained in $\mathbb C^2\ot V$, and so the face determined by
${\mathcal P}_0$ contains already infinitely many extreme points by Theorem \ref{2xn_thm} (i). This shows that ${\mathcal P}$
must be in general position.
\end{proof}

In the case of $2\ot 2$, all pairs of subspaces giving nontrivial faces are classified in \cite{ha_kye_04}.
The possible pairs $(\dim D,\dim E)$ of natural numbers arising from non-trivial faces are listed by
$$
(1,1),\quad (2,2),\quad (3,3),\quad (3,4),\quad (4,3),\quad (4,4).
$$
If we take an interior point $\varrho$ of the face $\tau(D,E)$ then these numbers are nothing but $(\rk\varrho,\rk\varrho^\Gamma)$.
It is clear that the $(1,1)$ case gives rise to the $0$-simplices consisting of single points. For the $(2,2)$ case,
we take an interior point $\varrho=\sum_{i=1}^2 |z_i\rangle\langle z_i|$ with product vectors $|z_1\rangle,|z_2\rangle$. Then we see that
this face is simplicial if and only if $\{|z_i\rangle\}$ is in general position. If this is the case, then we have  $1$-simplices.
We show that these are all possible cases. To do this,
suppose that a state $\varrho$ with expression (\ref{sep}) is an interior point of a simplicial face $\tau(D,E)$.
The case $(3,3)$ is split into two subcases:
either both $D^\perp$ and $E^\perp$ are spanned by product vectors $x\ot y$ and $\bar x\ot y$, respectively,
or both of them are spanned by non-product vectors.
In the first case, it is clear that there are infinitely many extreme points. In the second case,
we note that three product vectors are in general position if and only if they form a generalized
unextendible product basis. Therefore, we see that there are also infinitely many extreme points in this case
by Theorem \ref{2xn_thm} (i). In the $(3,4)$ case, $D^\perp$ is spanned by a product vector, and we see that there are infinitely many
product vectors in $D$ by Lemma 10 of \cite{2xn}.
All of them give rise to extreme points of the face of $\tau(D,E)$.

We turn our attention to the $2\ot 3$ case.
In this case, we have the following possibilities:
$$
(1,1),\quad (2,2),\quad (3,3),\quad (3,4), \quad (4,4), \quad (4,5), \quad (4,6),\quad (5,5),\quad (5,6),\quad (6,6).
$$
by \cite{choi_kye}. Here, we list up the cases $(p,q)$ with $p\le q$ by the symmetry.
Suppose that $\varrho$ of the form (\ref{sep}) determines a simplicial face isomorphic to the $k$-simplex. If $k\le 3$, then
this face must be of type $(k,k)$ by the above list. We next consider the case $k=4$. In this case,
$\{|x_i\otimes y_i\rangle:i=1,2,3,4\}$ is in general position by Proposition \ref{cor-2xn--}, and so
their partial conjugates are also in general position. Therefore, both $\{|x_i\otimes y_i\rangle \}$
and $\{|\bar x_i\otimes y_i\rangle \}$ are  linearly independent by
Proposition \ref{prop:prod_vec}.
This show that the face must be of type $(4,4)$, and a face of type $(3,4)$ is never a simplicial face.
In most $(4,4)$ cases,
they actually determines a simplicial faces isomorphic to the $3$-simplex by Theorem \ref{2xn_thm} (ii).
For a face $\tau(D,E)$ of type $(4,5)$, we consider product vectors
$(0,1)^\ttt\ot (0,0,1)^\ttt$ and $(1,z)^\ttt\ot (1,z,z^2)^\ttt$ for
$z=0,1,\omega,\omega^2$, where $\omega$ is the third root of unity.
It is simple to check that these five product vectors satisfy the
conditions in Proposition \ref{prop:induced_face}. Therefore, we
have a simplicial face isomorphic to $\Delta_5$.
For the $(4,6)$ and $(5,6)$ cases, we also refer to Lemma 10 of \cite{2xn} to see that there are infinitely many extreme points.
In a recent preprint \cite{chen_dj_2xd},
it was shown in Lemma 10 that if $\dim D+\dim E>3n$ for subspaces $D$ and $E$ of $\mathbb C^2\ot\mathbb C^n$
then there are infinitely many product vectors in $D$ whose partial conjugates belong to $E$. Therefore,
it is a simple consequence of this result that there is no simplicial face of type $(5,5)$.

\section{Conclusion}

We analyzed the convex geometry of the convex set $\mathbb S$ consisting of all $3\ot 3$ separable states to compare
with the convex set $\mathbb T$ consisting of PPT states.
With this geometric reasoning, we could show that
every PPT entangled edge state of rank four arises as the difference of two separable states of rank five.
For the higher dimensional cases, we recall that
the maximal dimension of subspaces in $\mathbb C^m\ot\mathbb C^n$ without product vector is given by
$$
p=(m-1)\times (n-1),
$$
and generic $(p+1)$-dimensional subspaces contain exactly
$\binom{m+n-2}{n-1}$ numbers of product vectors up to scalar multiplications.
See \cite{bhat,kiem,part,walgate,wall}.
In the $3\ot 3$ case, we exploited the strict inequality
$$
p+1=5<6=\binom{m+n-2}{n-1}
$$
to get the picture shown in Figure 1.
In the $3\ot 4$ case, we have $p+1=7$ and $\binom{m+n-2}{n-1}=10$. If these ten product vectors are in general position
then they give rise to a simplicial face isomorphic to the $9$-simplex by Proposition \ref{prop:simplex}. Six of them must be
linearly independent by Proposition \ref{prop:prod_vec}, but we could not determine if seven of them are linearly independent
in general or not.
If $n=4$, then the the number $\binom{m+n-2}{n-1}=\binom {m+2}3$ exceeds eventually the whole dimension $(4m)^2$. Therefore,
we cannot expect linear independence of corresponding product states as in Proposition \ref{prop:simplex}.
It would be interesting to know to what extent our approach may work.
We also have constructed separable states with unique decompositions, but with asymmetric ranks.
As an application, we constructed an explicit example of a $3\otimes 3$ PPT entangled state of type $(9,5)$.

In the $2\ot n$ cases, the situations go in a very different way, because
$$
p+1=n=\binom{m+n-2}{n-1}
$$
for these cases. Even though we found
simplicial faces with $n+1$ vertices, these do not give the pictures as in Figure 1, because the corresponding product vectors
are already linearly independent. This may explain why it is very difficult to construct PPT entangled edge states of various types
in the $2\ot n$ cases. In fact, it is still open if there exists a $2\ot 4$ PPT entangled edge state of type $(6,6)$.
See \cite{kye-prod-vec,kye_osaka}.

\end{document}